\documentclass[a4paper]{llncs}
 
\usepackage{graphicx}
\usepackage{amsmath, amssymb, amsfonts}
\usepackage[hidelinks]{hyperref}

\usepackage{xspace}
\usepackage{nicefrac}
\newcommand{\abs}[1]{\ensuremath{|#1|}}
\newcommand{\bigO}[1]{\ensuremath{O\left( #1 \right)}\xspace}
\newcommand{\degree}{\ensuremath{^\circ}}

\begin{document}

\title{Optimal online escape path against a certificate\thanks{Partially supported by the National Science Foundation, NSF grant CCF 1017539; see also \cite{k-pc-15}.}~\thanks{A preliminary version of this paper has been presented at SWAT.}}
\author{Elmar Langetepe\inst{1} \and David K\"ubel\inst{1}}
\institute{University of Bonn, Institute of Computer Science I, 53117 Bonn, Germany
\email{elmar.langetepe@cs.uni-bonn.de, kuebel@cs.uni-bonn.de}}
\titlerunning{Optimal online escape path}
\authorrunning{E. Langetepe and D. K\"ubel}

\maketitle

\begin{abstract}
More than fifty years ago, Bellman asked for the best escape path within  
a known forest but for an unknown starting position. This deterministic finite 
path is the shortest path that leads out of a given environment from any starting 
point. There are some worst case positions where the full path length is required. 
Up to now such a fixed  \emph{ultimate optimal escape path} for a known shape 
for any starting position is only known for some 
special convex shapes (i.e., circles, strips of a given width, 
fat convex bodies, some isosceles triangles).

Therefore, we introduce a different, simple and intuitive 
escape path, the so-called \emph{certificate path}.
This escape path depends on the starting position~$s$ and takes the distances from~$s$ to the 
outer boundary of the environment into account. Due to the additional information, 
the certificate path \emph{always} (for any position $s$) leaves the environment earlier than the ultimate escape path, in the above convex examples.

Next we assume that fewer information is available.
Neither the precise shape of the environment, nor the location of the starting point is known. 
For a class of environments (convex shapes and shapes with kernel positions), we design an \emph{online} strategy that always leaves the environment. 
 We show that the path length for leaving the environment is 
 always shorter than $3.318764$ the length of the corresponding certificate path.
We also give a lower bound of $3.313126$, which
shows that for the above class of environments the factor $3.318764$ is (almost) tight.
\end{abstract}

\section{Introduction\label{intro-sect}}
We consider the following motion planning task. 
Let us assume that we are given a simple polygon $P$ and a starting point $s$ inside $P$.  
We would like to design a simple path starting 
at  $s$ that finally hits the boundary and leaves the polygon. In the sense of a game, we can 
choose a path but then an adversary can rotate the polygon $P$ around  $s$ so that the path 
will leave the polygon very late. 

First, we assume that the distance from $s$ to the boundary is given into every direction.
We can apply a simple and intuitive strategy.
The \emph{certificate path} is the best combination of a line segment $l$ and 
an arc of length $l\alpha$ along the circle of radius $l$ around the starting point. 
So this path simply checks an angular portion of the environment for a distance~$l$. 
For a given starting point the certificate path is the best (shortest) such path that guarantees to hit the boundary. 
 Altogether the certificate path is a very simple \emph{escape path} for given $s$ and $P$  
 (if an adversary can only rotate $P$ around $s$).

In turn, for any given unknown starting position $s$ inside an unknown polygon, 
we would like to design an \emph{online strategy} (based on fewer information) that is never much worse than 
the length of the above certificate path. In this paper, we show that for a class of environments, 
there is a spiral strategy that leaves any such polygon and approximates the length of the certificate 
path within a ratio of $3.318674$. We also prove that this is an (almost) tight bound. 
There is no other strategy that always attains a better ratio against the length of the certificate path. 

This optimal online approximation is restricted to the following class of environments. 
We assume that  in any direction  from the unknown starting point only one 
boundary point exists. The distance to the boundary points still remains unknown. 
This subsumes any unknown convex environment (for any unknown starting position) and also
unknown star-shaped environments (for any unknown starting point inside the kernel). 
The motivation of comparing an online escape path for special polygons 
(star-shaped) and special starting positions (inside the kernel) with 
a path that is computed with some additional but not complete information (certificate path) 
stems from the following observation. 

For a known polygonal shape and an unknown starting point, it is possible to 
define an \emph{ultimate optimal escape path}. This path will lead out of the environment 
for any starting point  and any rotation of the polygon. The ultimate optimal escape path is the 
shortest finite path with this property. 
The clue is that only the polygon is known but neither the starting position, nor 
the rotation around the starting position. 
The path is motivated by the situation of swimming in the fog in a 
pool of known shape. As it is foggy, the starting point and the rotation around the 
starting point is not known. 
Unfortunately, ultimate optimal escape paths have been found only for a 
few special convex shapes 
(circles, strips of given width, fat convex bodies, isosceles triangles, $\ldots$). 
It is unrealistic to think that such paths will be found for more complicated convex 
or star-shaped environments. 

Fortunately, for the few cases where an ultimate optimal escape path is known, the certificate path is not only a good approximation.
We can even show that the certificate path \emph{beats} 
the ultimate escape paths for \emph{any} starting point in these examples. 
Therefore, we are convinced that the certificate path can serve as a substitute for the unknown ultimate optimal escape path in the restricted cases.

The paper is organized as follows. In the next section we present related work. 
The certificate path is introduced and defined in Section~\ref{cert-sect}.
Different justification for the measure is discussed in Section~\ref{just-sect}.
Finally, in Section~\ref{onlappe-sect}, we present and analyse a strategy
with path length not larger than  $3.318674$ times the length of 
the certificate path. 
The strategy is a logarithmic spiral attained by keeping aware of 
two extremes of the certificate. 
Optimizing the spiral for two extremes is also different from classical 
logarithmic spiral constructions where we normally optimize against 
a single distance (shortest path). 
In Section~\ref{opt-sect}, we present a general lower bound which proves that the given strategy is 
almost optimal for the restricted cases. No other strategy 
will have a better ratio than $3.313126$ against the length of the certificate 
path. 
Proving lower bounds is a tedious task, the construction and the analysis might 
be interesting in its own right.  

\section{Related work\label{related-sect}}
 The \emph{Swimming-in-the-fog} problem 
is a game where  two players, a searcher and a hider, 
compete with each other. The searcher tries to reach the boundary of a known shape
from its starting point along a single finite path, 
while  the hider can rotate and translate the environment  
so that the path of the searcher will cross the boundary as late as possible.
For a given shape, the shortest finite path that always leads 
out of the given environment can be denoted as an \emph{ultimate optimal escape path} 
as mentioned before. 

Since the first work by Koopman in 1946, search games have been studied in many variations in the last
60 years.
The book by Gal \cite{g-sg-80} and the reissue by Alpern and Gal
\cite{ag-tsgr-03} gives a comprehensive overview of such search
game problems, also for unknown environments. 

The above problem goes at least back to 1956 and to Bellman~\cite{b-mp-56}, who similarly 
asked 
for the shortest escape path within a known forest but for an unknown starting point. Since 
then, the problem has attracted a lot of attention. Unfortunately until today, 
the problem could be solved only  
for very special convex environments (circles, strips of given width,
rectangles, fat convex bodies, isosceles triangles); see for example the 
monograph of Finch and Wetzel~\cite{fw-lf-04}.

For circles  and fat convex bodies, it was shown that the diameter is the 
ultimate optimal escape path; see Finch and Wetzel~\cite{fw-lf-04}.
For the infinite strip of width $l$, the ultimate optimal escape path
is due to Zalgaller~\cite{z-hgw-61,z-qb-05}.
For the simple equilateral triangle of side length $1$, the zig-zag path of Besicovitch~\cite{b-accoe-65}
of length $\approx0.981981$ is optimal; see also \cite{cm-btcua-06}. 
Furthermore, in 1961 Gluss~\cite{g-mpsc-61} introduced the problem of searching 
for a circle $C$ of given radius $s$ and given distance $r$ away from the start~$A$. 
Two different cases can be considered, either $A$ is inside $C$ or not. 
Interestingly, in the latter case and for $s=1$ a certificate path with length $l=r$ and an arc 
of length $2\pi\cdot l$ is the best one can do. 

It is unrealistic to think that such ultimate optimal escape paths will be found for
more complicated environments. 
As an alternative, we introduced a simple and natural \emph{certificate path} for a known environment and a given starting point~$s$.
It is computed individually for any starting point and takes the distance 
distribution from~$s$ to the boundary into account.
Fortunately, for the cases where ultimate optimal escape paths are known, they are always outperformed by the certificate path for any possible starting point.
This means we can prove that the certificate path always leaves the environment earlier.

The use of alternative comparison measures has some tradition. 
For example for the problem of searching for a point in a polygon and 
competing against the shortest path, there is no competitive strategy. 
For this reason other comparison measures have been suggested in this case; 
see Fleischer et al.~\cite{fkklt-colao-08} or Koutsoupias et al.~\cite{kpy-sfg-96}.  
Additionally, comparing the online strategy to the shortest path to the boundary is often a very 
difficult task. For example, the spiral conjecture for searching for 
a single line or a single ray against the shortest path is still open. In this sense 
our result might be considered as an intermediate step.

For the design of the online escape path, we assume that neither the precise shape of the environment, nor the position of the starting point is known.
We make use of the competitive framework to show that the online strategy can compete with the certificate path, which is computed with more information.
That is, we compare the length of
the online escape path from a starting point to the  boundary to the length of the certificate path 
to the boundary computed for the known environment and starting point. 
The competitive framework was introduced by Sleator and Tarjan
\cite{st-aelup-85}, and used in many settings since then; see for example
the survey by Fiat and Woeginger \cite{fw-ola-98} or, for the field
of online robot motion planning, see the surveys
\cite{ikkl-ccnt-02,rksi-rnut-93}.

Our optimal online approximation is restricted to the following class of environments. 
We assume that  in any direction  from the unknown starting point, only one boundary point exists. 
The distance to the boundary points remains unknown. 
This subsumes any unknown convex environment (for any unknown starting position) and also
unknown star-shaped environments (for any unknown starting point inside the kernel). 
In this sense, the certificate is  also a  natural extension of the discrete performance measure Kirkpatrick~\cite{k-hd-09} 
mentioned in the discrete case of
 searching for the end of a set of $m$ given  lists of unknown length. In his setting 
 it is sufficient to reach the end of only one list. 
In our configuration, this means that we have exactly $2\pi$ directions of 
unknown distance and it is sufficient to reach the shoreline in a single point. 
The corresponding relationship is shown in Section~\ref{just-sect}.

We will see that our solution is a specific logarithmic spiral. 
In general, logarithmic spirals are natural candidates for optimal 
competitive search strategies, but in almost all cases the optimality remains a conjecture; 
see \cite{bcr-sp-93,efkklt-colsr-06,f-lsc-05,fz-ss-05,g-sg-80}. 
In \cite{l-oss-10} the optimality of spiral search was shown for 
searching a point in the plane with a radar. Many other conjectures are still  open. 
For example Finch and Zhu~\cite{fz-ss-05} considered the 
problem of searching for a line in the plane.
\emph{The relevant conjecture 
that the family of logarithmic spirals contains the minimal path remains open}.

 \section{The certificate path}\label{cert-sect}
%
\begin{figure}
\begin{center}
\includegraphics[scale=0.28]{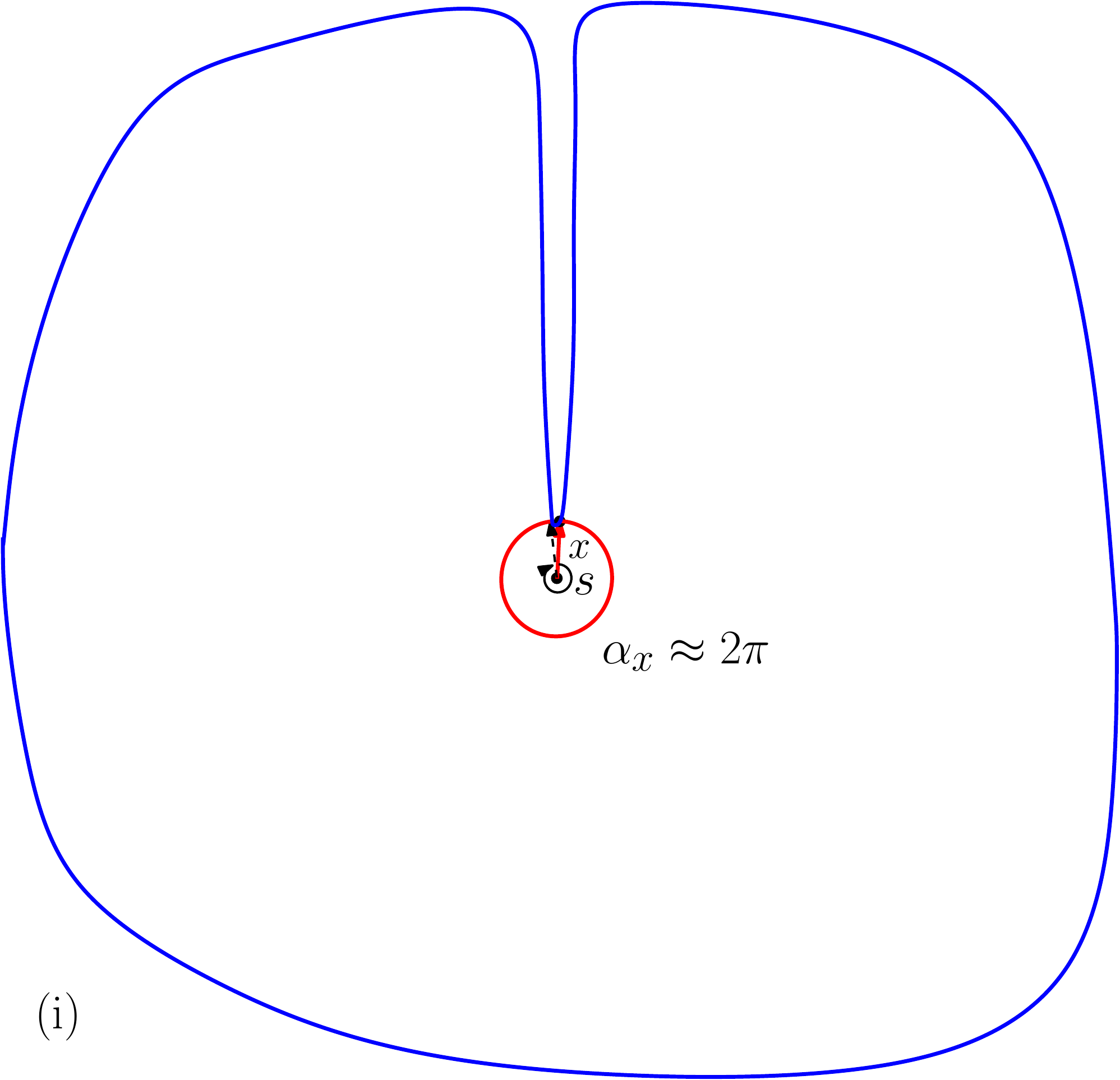}\quad\quad\quad
\includegraphics[scale=0.32]{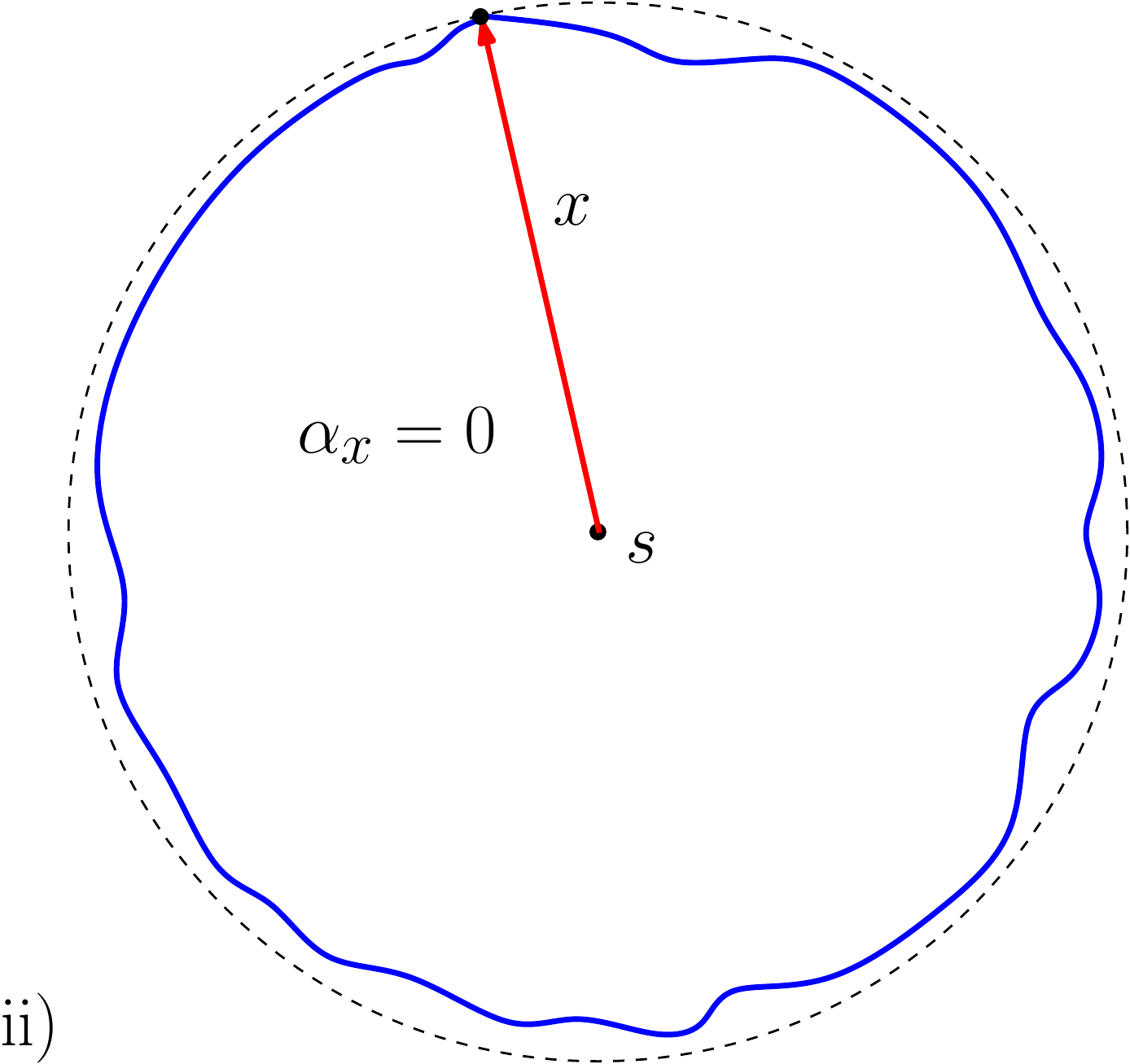}
\caption{
Two extreme situation for reaching the boundary with a circular arc.
(i) The radial maximal distance from $s$ to the boundary is almost the same in any direction.
It suffices to move in an arbitrary direction of maximal distance, which is optimal. 
(ii) The distance to some few boundary points is very small, while large to most of the others.
A reasonable path checks the small distance with a circular arc of length approximately $2\pi$.  
In both cases, $x(1+\alpha_x)$ is minimal among all such circular strategies.
}
\label{ExtremeCirc-fig}
\end{center}
\end{figure}
Assume that you are located in an unknown environment and would like to reach its boundary. 
Formally, for the environment we consider a closed Jordan curve $B$ 
that subdivides the Euclidean plane into exactly two regions.
The starting point $s$ lies inside the inner region, say $P$.
The task is to reach a point on the boundary $B$, as soon as possible. 

If you have some idea about the distance $x$ from $s$ to the boundary $B$ but nothing more, 
it is very intuitive to move along the circle of radius $x$ around the starting point. 
Therefore, a reasonable strategy moves toward this circle along a shortest path (by radius $x$) 
in some direction and then follows the circle in either clockwise or counterclockwise direction 
until the boundary is met.
Let us denote this behaviour a \emph{circular strategy}.
If we hit the boundary after moving an arc $\alpha_x$ along the circle, the overall path 
length is given by $x(1+\alpha_x)$.  

We would like to use such a circular strategy of small path length. 
In the sense of a game, the adversary can only rotate the environment around the 
starting point and the certificate path guarantees to hit the boundary for any rotation. 

\subsection{Extreme cases and general definition}\label{extremedef-sect}
Let us first consider two somehow extreme examples of the above intuitive idea as 
given in Figure~\ref{ExtremeCirc-fig}.
If the distance from $s$ to the boundary is almost the same in any direction 
(similar to a circle),  a line segment with maximal distance to the boundary (roughly the radius of the circle) 
will always hit the boundary and is indeed a very good escape path for any direction; 
see Figure~\ref{ExtremeCirc-fig}(ii). 
The movement along an arc is not necessary in this case.
In other words, $\alpha_x$ equals $0$. We check a single direction for the largest distance. 

On the other hand, it might be the case that the distance to the boundary is very large w.r.t. almost
all, but small (distance $x$) for some few directions from $s$.
Then, a segment of length $x$ and a circular arc of length $x\alpha_x$ with $\alpha_x\approx 2\pi$
will hit the boundary for any starting direction of the segment $x$; see Figure \ref{ExtremeCirc-fig}(i). 
The overall path length $x(1+\alpha_x)$ is comparatively small. 
The certificate path checks a small distance for many (almost all) directions. 

Now, consider a more general environment modelled by a simple polygon $P$ and a fixed 
starting point $s$ in $P$ as given in Figure~\ref{exampleStarCert-fig}(ii). 
For convenience, we make use of an example, where any boundary  point $b$ of $P$ is 
\emph{visible} from $s$, i.e. the segment $sb$ lies fully inside $P$.
Or, the other way round, $s$ lies inside the kernel of $P$. 

\begin{figure}
\begin{center}
\includegraphics[width=\textwidth]{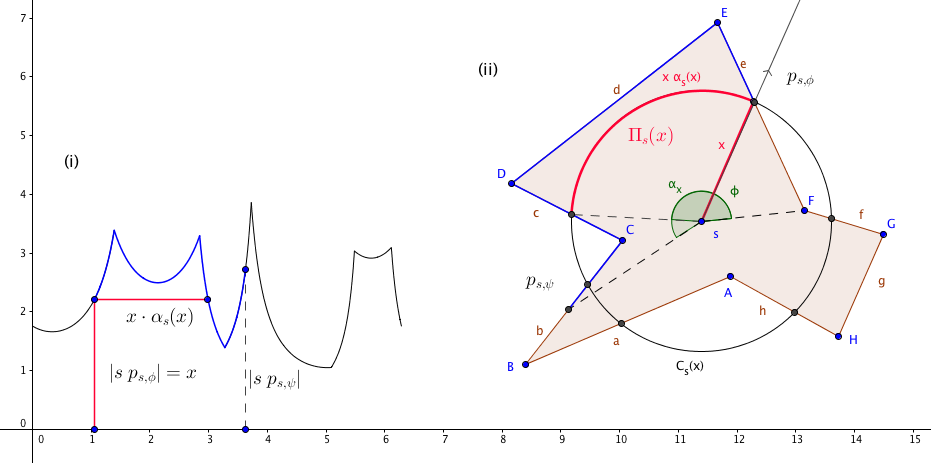}
\caption{ (ii) Consider the  polygon $P$ and a starting point $s$. Let us assume that we 
radially sweep  the boundary of $P$ (starting from point $F$ 
with angle $0$) in counter-clockwise order and calculate the distance from 
the boundary to $s$ for any angle.  
(i) shows this \emph{radial distance function} of the boundary of $P$ from $s$ in polar 
coordinates for the interval $[0,2\pi]$. The blue sub-curve corresponds to the blue boundary part in (ii). 
The certificate path $\Pi_s(x)$ for distance $x$ is the longest path that successfully checks the distance 
$x$ by a circular strategy. This means that it hits the boundary for any starting 
direction $\phi$ of $x$ in $P$.  
In the polar-coordinate setting in (i) this is a path with two line segments 
of length $x$ and $\alpha_x$ that always 
hits the boundary of the radial distance function independent from the starting angle $\phi$. 
}
\label{exampleStarCert-fig}
\end{center}
\end{figure}
For the polygon $P$ and for any radial direction $\phi\in [0,2\pi]$ from $s$, we consider 
the boundary point $p_{s,\phi}$ on $P$ in direction $\phi$. This gives a \emph{radial distance function} 
$f(\phi):=|s\;p_{s,\phi}|$, as depicted in Figure~\ref{exampleStarCert-fig}(i). 

Now, let $p_{s,\phi}$ be a point with  distance $x:=|s\;  p_{s,\phi}|$ in direction $\phi$.
For any circle $C_s(x)$ with radius $x$ around $s$ such that $C_s(x)$ hits the boundary 
of $P$, there will be some maximal arc $\alpha_s(x)$ so that the above simple circular 
strategy is successful.
Note, that this is independent from the starting direction for $x$.
We are looking for the maximum circle segment of $C_s(x)$ that fully lies inside $P$. 

Let $\Pi_s(x)$ denote the \emph{certificate path for distance $x$} of maximal  
length $x(1+\alpha_s(x))$. 
The interpretation is that this finite path will always touch the boundary, independent from the starting direction for $x$.
The adversary can only rotate the environment 
in order to attain a worst case length of $x(1+\alpha_x)$.

Every certificate path for a distance $x$ corresponds to two connected segments in the plot of the radial distance function ; see Figure \ref{exampleStarCert-fig}(i).
The vertical segment of length $x$ represents the radius and the horizontal segment of length $\alpha_x$ represents the arc of the circular strategy.
For any starting angle, this path will touch the boundary of the distance function.

We define he overall \emph{certificate path} $\Pi_s$ in $P$ for a given starting point $s$ 
as the shortest certificate path $\Pi_s(x)$ among all distances $x$.
That is, the certificate for $P$ and $s$ is:
$$\Pi_s:= \min_x \Pi_s(x)= \min_x x(1+\alpha_s(x))\;.$$ 
For both extreme situations in Figure~\ref{ExtremeCirc-fig}, the presented paths 
equal the overall certificate paths for the given environments. 

Finally, consider the case when parts of the boundary are not visible from $s$.
In this case the radial maximal distance computation is no longer 
a function, it results in a curve; see Figure~\ref{exampleGenCert-fig}. 
However, the certificate $\Pi_s(x)=x(1+\alpha_x)$ for distance $x$,
the maximal arc $\alpha_x$ and the overall certificate $\Pi_s$ are still well defined. 
Note that, for the certificate $\Pi_s$ in any polygon $P$ and 
the corresponding arc $\alpha_x$, we can conclude  $\alpha_x\in [0,2\pi]$. 
This holds since the shortest distance $d_s$ from $s$ to the boundary 
always results in a candidate $d_s(1+2\pi)$. All other reasonable distances $x$ 
are larger than $d_s$ and $\alpha_x\leq 2\pi$ holds for the optimal $x$. 

The online approximation of the certificate path of an arbitrary 
unknown polygon with a spiral strategy cannot be competitive in general. 
The corresponding ratio could grow arbitrarily large, as the polygon might wind itself around the spiral.
In more general environments other online strategies have to be applied. 
For example, one could think of a connected sequence of circles $C_i$ with exponentially increasing radii $r^i$.
This would give at least a constant competitive ratio.
However, obtaining the optimal strategy for such cases might be hopeless.

\begin{figure}[th]
\begin{center}
\includegraphics[width=\textwidth]{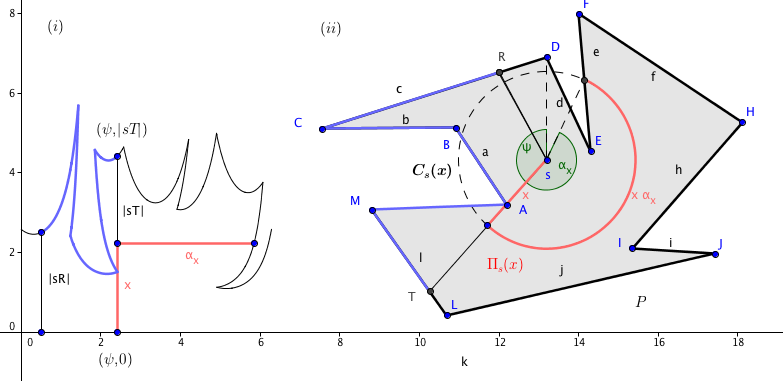}
\caption{ (ii) For a general polygon $P$ and a point $s\in P$ not
the whole boundary might be visible. 
 (i) The radial distance computation appears to be a curve. 
 Nevertheless, the certificate  for distance $x$ and also the 
 overall certificate is well-defined and has the 
 same geometric interpretation. }
\label{exampleGenCert-fig}
\end{center}
\end{figure}

\subsection{Justification of the certificate\label{just-sect}}
The certificate path is an intuitive and simple way of leaving an environment.
It can be computed in polynomial time, as we show in \ref{appendix:polytimecomp-sect}.
We can interpret the certificate as a path that balances depth-first and breadth-first
search for the starting position $s$ in a way that the resulting path is as short as possible.
That way, it outperforms the ultimate optimal escape path at any given starting position 
for all known cases.

For circles, semi-circles with an opening angle $\alpha$ larger than 60 degrees and for \emph{fat} convex bodies, the ultimate optimal escape path equals the diameter; see Figure~\ref{DiameterEx-fig}.
\begin{figure}
\begin{center}
\includegraphics[scale=0.42]{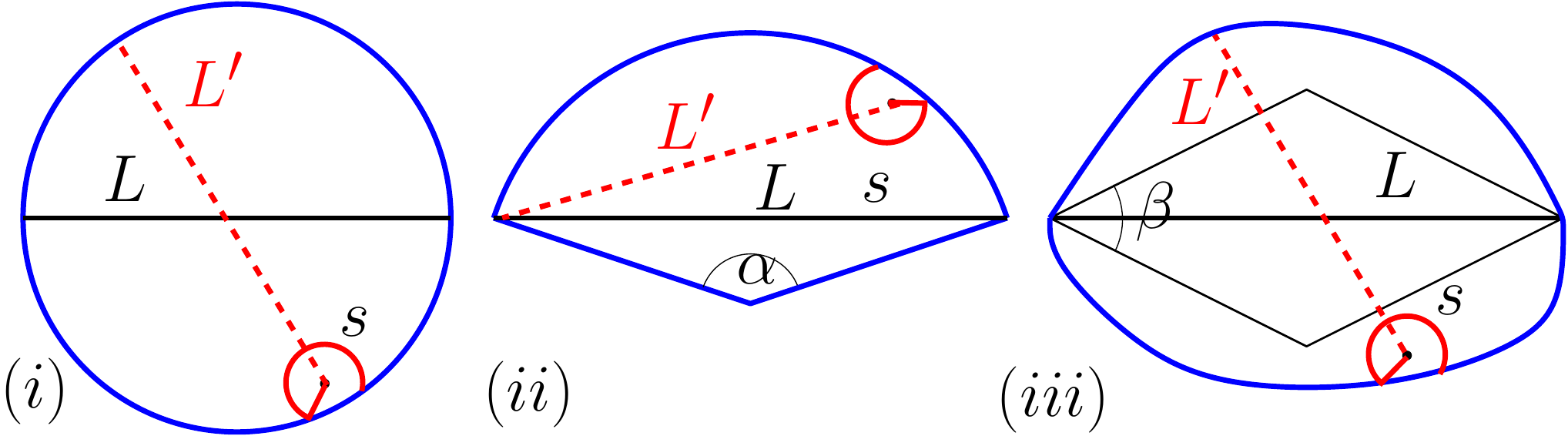}
\caption{Three environments, where the diameter $L$ is the ultimate optimal escape path. 
(i) A circle, (ii) a semi-circle with angle $\alpha\geq 60^\circ$. (iii) A so-called 
\emph{fat} convex body.
An environment is called fat if a rhombus with angle $\beta$ of least $60^\circ$ fits into it.
Then, the diameter of the rhombus equals the diameter of the environment.
The usage of the ultimate optimal escape path (line segment $L$) from $s$ results in the largest distance $x$ 
to the boundary (dashed path $L'$) in the worst case. The certificate path is as least as good as the diameter path. 
}
\label{DiameterEx-fig}
\end{center}
\end{figure}
For any position $s$ the worst case for this escape path is given by a rotation of the environment so that a segment of maximal length $x$ is required.
As the certificate path $\Pi_s$ considers such a path as a possible alternative, the certificate path is as least as good as the diameter for any position.

For the case when the environment is an equilateral triangle or an infinite strip, the certificate outperforms the ultimate optimal escape path for any starting position in the worst case.
As this is not straight forward to see, we give small proofs in the following.

\paragraph{The Equilateral Triangle \& Besicovitch's path.}
Consider an equilateral triangle as depicted in Figure~\ref{ExampleBesi-fig}.
Besicovitch's zig-zag path is the ultimate optimal escape path for this environment; c.f. 
the discussion by Besicovitch~\cite{b-accoe-65} and the proof of 
optimality by Coulton and Movshovich~\cite{cm-btcua-06}.
The zig-zag path is symmetric and consist of three segments of length $\sqrt{\nicefrac{3}{28}}$ each.
This constitutes a total length of $\approx 0.981918$.
An example of a worst case starting point $X_1$ is given in Figure~\ref{ExampleBesi-fig}. 

Consider the following observation.
For starting points somewhere in the center of the triangle, the zig-zag path is worse than 
the largest distance to the boundary;
see for example starting point $s_1$ in Figure~\ref{ExampleBesi-fig}. Thus, the 
certificate is shorter for these points. For starting points close to the boundary 
the certificate is significantly better by a short circular check; see for example starting point $s$ in Figure~\ref{ExampleBesi-fig}.
Now, we give a formal proof, that the zig-zag path can never beat the certificate.
The values for $\alpha=\arcsin\left(\nicefrac{1}{\sqrt{28}} \right)\approx 10.9^\circ$ and $x=\sqrt{\nicefrac{3}{28}}$ are due to Coulton and Movshovich~\cite{cm-btcua-06}.

If Besicovitch's zig-zag path does not hit a point on the boundary with maximum distance away 
from the starting point, it does not hit one of the three vertices of the triangle. 
Only in this case  Besicovitch's zig-zag path can beat the certificate. 
In this case, the Besicovitch path has to make use of at least two segments (each of 
length $x=\sqrt{\nicefrac{3}{28}}$) 
for leaving the triangle. As shown in Figure~\ref{BesiAgainstCert-fig}, only 
a small area close to a vertex of the triangle has to be checked. 
For all those points, we find a circular strategy that is shorter than $2x$, which is always 
required by the zig-zag path.  
Note that we can easily extend this argumentation to the 
family of isosceles triangles analysed by Coulton and Movshovich~\cite{cm-btcua-06,m-bte-11}.  

\begin{figure}
\begin{center}
\includegraphics[width=\textwidth]{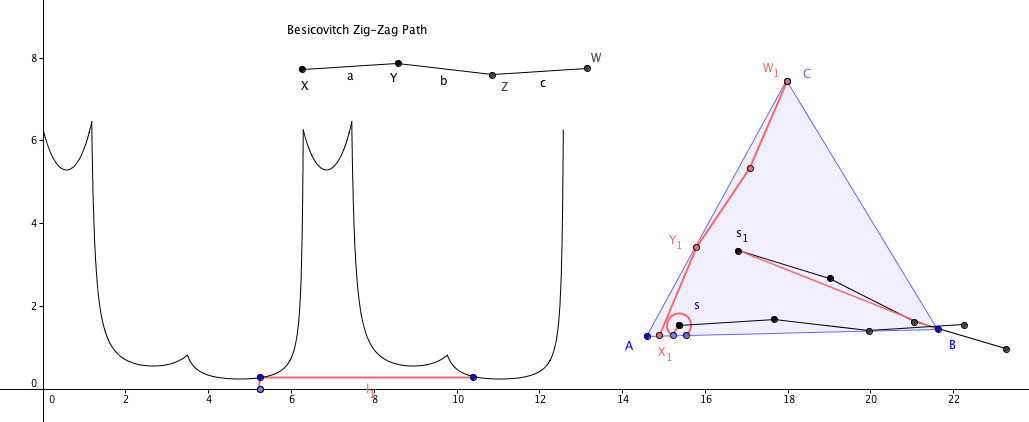}
\caption{Besicovitch's zig-zag path is the ultimate optimal escape path for an equilateral triangle.
A worst case position is given by $X_1$ for example. 
The plot on the left-hand side shows the radial distance curve of point $s$.
The certificate path, for such points close to the boundary, is very short.
In contrast to this, the usage of Besicovitch's path is much worse, as the worst case is attained when leaving the triangle on the boundary farthest away.}
\label{ExampleBesi-fig}
\end{center}
\end{figure}
\begin{figure}
\begin{center}
\includegraphics[scale=0.4]{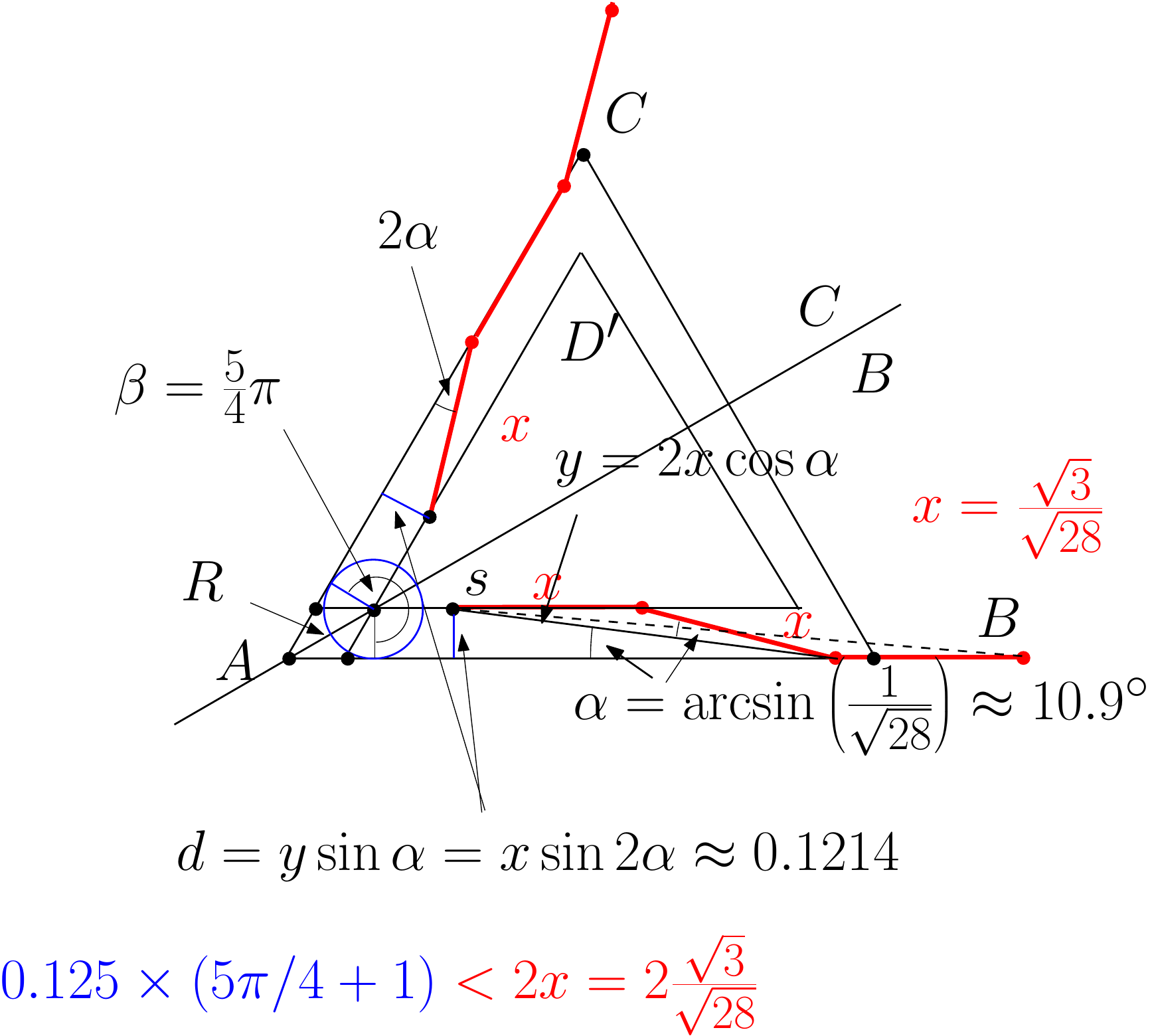}
\caption{For comparing the certificate and the Besicovitch's zig-zag path we only 
have to consider the case when the triangle cannot be rotated so that the
zig-zag path ends in the vertex farthest away. For starting points $s$ below 
the bisector of $A$ and $C$ this can only happen if $s$ is at least vertical distance 
$d=x \sin 2\alpha$ away from segment $AC$.  For starting points above the 
the bisector of $A$ and $C$ this can only happen, if $s$ is at least vertical distance 
$d=x \sin 2\alpha$ away from segment $AB$. 
For both cases it remains to consider the points in the rhomboid $R$. 
Fortunately, we can use  a circle of radius $d'=0.125$  (slightly larger than $d$) 
so that the circle of radius $d'$ with starting point in $R$ 
touches the boundary with an arc of length at most $2\pi-\frac{3}{4}\pi=\frac{5}{4}\pi$ and 
$d'(\frac{5}{4}\pi+1)$ is always strictly smaller than $2x$. } 
\label{BesiAgainstCert-fig}
\end{center}
\end{figure}

\paragraph{The Infinite Strip \& Zalgaller's path.}
\newcommand{\zalgaller}{\ensuremath{\zeta}\xspace}
\newcommand{\zalgPoint}[1]{Z\ensuremath{_{#1}}\xspace}
\newcommand{\width}{\ensuremath{l}\xspace}
\newcommand{\baseline}{x\xspace}
\newcommand{\distance}{\ensuremath{d}\xspace}
\newcommand{\zalgExit}{E\xspace}
\newcommand{\upperFoot}{U\xspace}
Consider the infinite strip of width \width.
Zalgaller's path \zalgaller is an ultimate optimal escape path;
see \cite{z-hgw-61}, \cite{z-qb-05} and an alternative proof in~\cite{cglq-cwors-03}.
It consists of four line segments and two arcs, which are defined by the following values; see Figure \ref{figure:zalgaller}.
\begin{equation*}
\begin{array}{rcl}
\phi & = & \arcsin \left( \frac{1}{6} + \frac{4}{3} \sin \left( \frac{1}{3} \arcsin \frac{17}{64} \right) \right)\\[1em]
\psi & = & \arctan \left( \frac{1}{2} \sec \phi \right)\\[1em]
x	 & = & \sec \phi \approx 1.043590
\end{array}
\end{equation*}
The total length from \zalgPoint{1} to the end at \zalgPoint{7} along \zalgaller is approximately 2.278292 \width.
The path is symmetric with regard to the bisector of the baseline.
The baseline itself has length $\baseline \cdot \width$.

\begin{figure}
	\begin{minipage}{0.48\textwidth}
		\includegraphics[width=\textwidth]{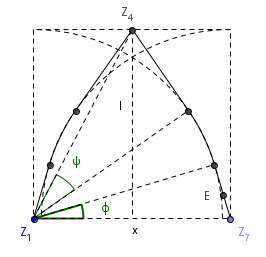}
	\end{minipage}
	\begin{minipage}{0.48\textwidth}
	\includegraphics[width=0.85\textwidth]{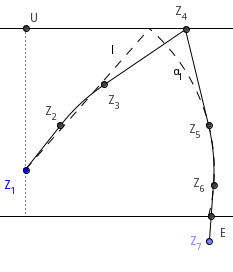}
	\end{minipage}
	\caption{The figure on the left-hand side shows the curve of \zalgaller.
	On the right-hand side, \zalgaller has been rotated around \zalgPoint{1} so that \zalgPoint{4} lies on the upper bound of the strip.
	Dashed lines indicate a configuration of the certificate path which is shorter than \zalgaller for the starting position \zalgPoint{1}.}
	\label{figure:zalgaller}
	\label{ExampleStrip-fig}
\end{figure}

To prove that the certificate path outperforms Zalgaller's escape path for any given starting point in the infinite strip, we proceed in two steps.
At first, we show that for any starting position escaping along \zalgaller takes at least $2.15\cdot \width$ in the worst case.
Then, we show that, in the worst case, the length of the certificate path is always below this bound.

For the first argument, we argue that \zalgaller can always be rotated around the starting point appropriately.
Figure \ref{figure:zalgaller} shows the shape of \zalgaller on the left-hand side.
W.l.o.g. we assume that we follow \zalgaller in clockwise orientation and name the vertices of \zalgaller from $Z_1$ to $Z_7$ appropriately.
We denote \distance the distance from the staring point $s$ (which equals $Z_1$) to the lower bound of the strip.
We may assume that $\distance\in[0;\nicefrac{\width}{2}]$, otherwise we can turn the whole configuration by 180\degree.
Now, we rotate \zalgaller around \zalgPoint{1} so that \zalgPoint{4} lies on the upper bound of the strip and \zalgPoint{3}, \zalgPoint{5} both lie inside the strip; see Figure \ref{figure:zalgaller} on the right-hand side.
This is always possible as the following argument shows.
As $\triangle \zalgPoint{3}\zalgPoint{4}\zalgPoint{5}$ is equal-sided, the segment $\left(\zalgPoint{3}, \zalgPoint{5}\right)$ is parallel to the baseline and $ \measuredangle\zalgPoint{4}\zalgPoint{5}\zalgPoint{1} = 90\degree $, we have $\measuredangle\zalgPoint{5}\zalgPoint{3}\zalgPoint{4} =  90\degree - 2 \phi$.
The congruency of $\triangle \zalgPoint{1}\zalgPoint{5}\zalgPoint{4}$ and $\triangle \zalgPoint{1}\zalgPoint{4}X$ allows to conclude $\measuredangle\zalgPoint{3}\zalgPoint{4}\zalgPoint{1} = 2 \phi - \psi$.
On the one hand, this proves that $\zalgPoint{3}$ indeed lies in the strip as $\measuredangle\zalgPoint{3}\zalgPoint{4}\zalgPoint{1} = 2 \phi - \psi < \arcsin\left( \distance \cdot (0.25 \baseline^2 + 1)^{-1/2} \right) = \measuredangle\upperFoot\zalgPoint{4}\zalgPoint{1}$ for any $\distance\in \left[0; \nicefrac{\width}{2} \right]$.
On the other hand, we can also show that $\zalgPoint{5}$ lies inside the strip as $\measuredangle \upperFoot\zalgPoint{4}\zalgPoint{5} < \measuredangle \upperFoot\zalgPoint{4}\zalgPoint{1} + \measuredangle \zalgPoint{3}\zalgPoint{4}\zalgPoint{5} = \arcsin\left( \distance \cdot (0.25 \baseline^2 + 1)^{-1/2} \right) + 4 \phi < 180\degree$ for any $\distance\in \left[0; \nicefrac{\width}{2} \right]$.

Now that we know that this configuration can be always be attained, we consider point \zalgExit, where \zalgaller leaves the strip.
Summing up the length of \zalgaller from \zalgPoint{1} to E exceeds the lower bound of $2.13$, which gives the proof.
Due to symmetry, the lower bound also holds if \zalgPoint{7} lies in the infinite strip and we follow \zalgaller counter-clockwise.

For the second argument, we consider certificate paths for two different distances, which perform better than \zalgaller for a given starting position.
Again, we denote \distance the distance from the starting position to the lower bound of the strip and assume w.l.o.g. that \distance $\in[0;\nicefrac{\width}{2}]$.
We split this interval in half.
For starting positions with $\distance\in[0,\nicefrac{\width}{4}]$, we know that $\Pi_s(\distance)$ is a certificate path and $\alpha_\distance$ is a full circle.
In this case, the length of $\Pi_s$ is well below $2.11 \width$.
Otherwise, we have $\distance\in [\nicefrac{\width}{4},\nicefrac{\width}{2}]$.
In this case, we consider the certificate path for distance \width.
The length of $\Pi_s(\width)$ is strictly decreasing with growing $\distance\in[\nicefrac{\width}{4},\nicefrac{\width}{2}]$.
Consequently, the length of $\Pi_s(\width)$ has a maximum value for $\distance = \nicefrac{\width}{2}$ and the length of $\Pi_s(\width)$ is slightly below $2.11 \width$.

We saw that for any starting point in the strip, there are certificate paths for certain distances that are shorter than $2.11 \width$.
As the certificate path is the overall minimum it is below this bound as well.
For any starting point, Zalgaller's path \zalgaller can be rotated so that leaving the strip takes at most $2.15 \width$.
Consequently, the certificate path outperforms \zalgaller for any starting point.

\paragraph{The relation to breadth \& depth first search.}
Finally, we would like to relate the certificate to a discrete cost measure Kirkpatrick 
introduces in \cite{k-hd-09}. He analyses the problem of digging for oil at $m$ different 
locations $s_i$, where $|s_i|$ denotes the (unknown) distance to the source of the oil at 
the corresponding location. In this scenario, no extra costs arise for switching the location.
The challenge is to find a strategy that reaches one source of oil while assuring a small 
overall digging effort.

At first, Kirkpatrick considers (partially informed) strategies.
Those are given all distances from the top to the sources of oil, but not the corresponding location:
In case the distances \abs{s_i} have about the same length at all locations, he states that a 
depth-first searching strategy is certainly effective. Thus, a single location can be chosen 
for digging, as Figure~\ref{CorrFind-fig}(i) indicates.
Although at the chosen location, the distance to the source might be greatest, the 
digging costs are almost optimal. In case the distance to the source of a single location is 
significantly shorter than all others, a breadth-first searching strategy performs best.
Figure~\ref{CorrFind-fig}(ii) shows that digging at every location with a certain effort $x$ 
still achieves a small overall effort of $x\cdot m$ in the worst case.
These two extreme situations are similar to the cases outlined in Section~\ref{extremedef-sect} 
and depicted in Figure~\ref{ExtremeCirc-fig}.
For the general case, Kirkpatrick suggests to use a hybrid strategy.
If $f_1\geq f_2 \geq \cdots \geq f_m$ denotes the sorted set of distances, he suggests to 
choose $i$ so that $i \cdot f_i$ is minimal.
The hybrid strategy digs at $i$ (arbitrarily chosen) locations up to the same depth $f_i$.
In the worst case, this strategy reaches a source at the last location with a final effort 
of $i \cdot f_i$; see Figure~\ref{CorrFind-fig}(iii). Among all such partially informed 
strategies, this hybrid strategy is certainly optimal and achieves a maximum digging 
effort of $\lambda := i \cdot f_i$. Similar to this hybrid strategy, we defined the certificate 
path in the previous section. The certificate path can also be considered as a mixture of 
depth-first and breadth-first searching. However, the certificate path models a motion.
The effort of the digging strategy to explore a certain depth depends on the product of the 
number chosen locations and the digging depth. In contrast to this, the effort of the certificate 
path depends on the sum of the searching depth and width.
Consequently, the certificate path is a stronger cost measure than the equivalent of 
the hybrid digging strategy in the plane.

During the further analysis, Kirkpatrick compares a totally uninformed digging strategy to 
the optimal hybrid strategy. He proves that this strategy approximates the hybrid 
strategy in \bigO{\lambda \log(\min(m,\lambda)}
and shows that this factor is tight. Similar to his approach, we compare the certificate path 
to a totally uninformed spiral strategy and obtain a constant competitive ratio. 
David Kirkpatrick~\cite{k-pc-15} brought up the question about what happens in a continuous setting. 
Note that the game is a quite different in this case, as we take movements
in the plane into account and also require a starting orientation. 
\begin{figure}[ht]
\begin{center}
\includegraphics[scale=0.33]{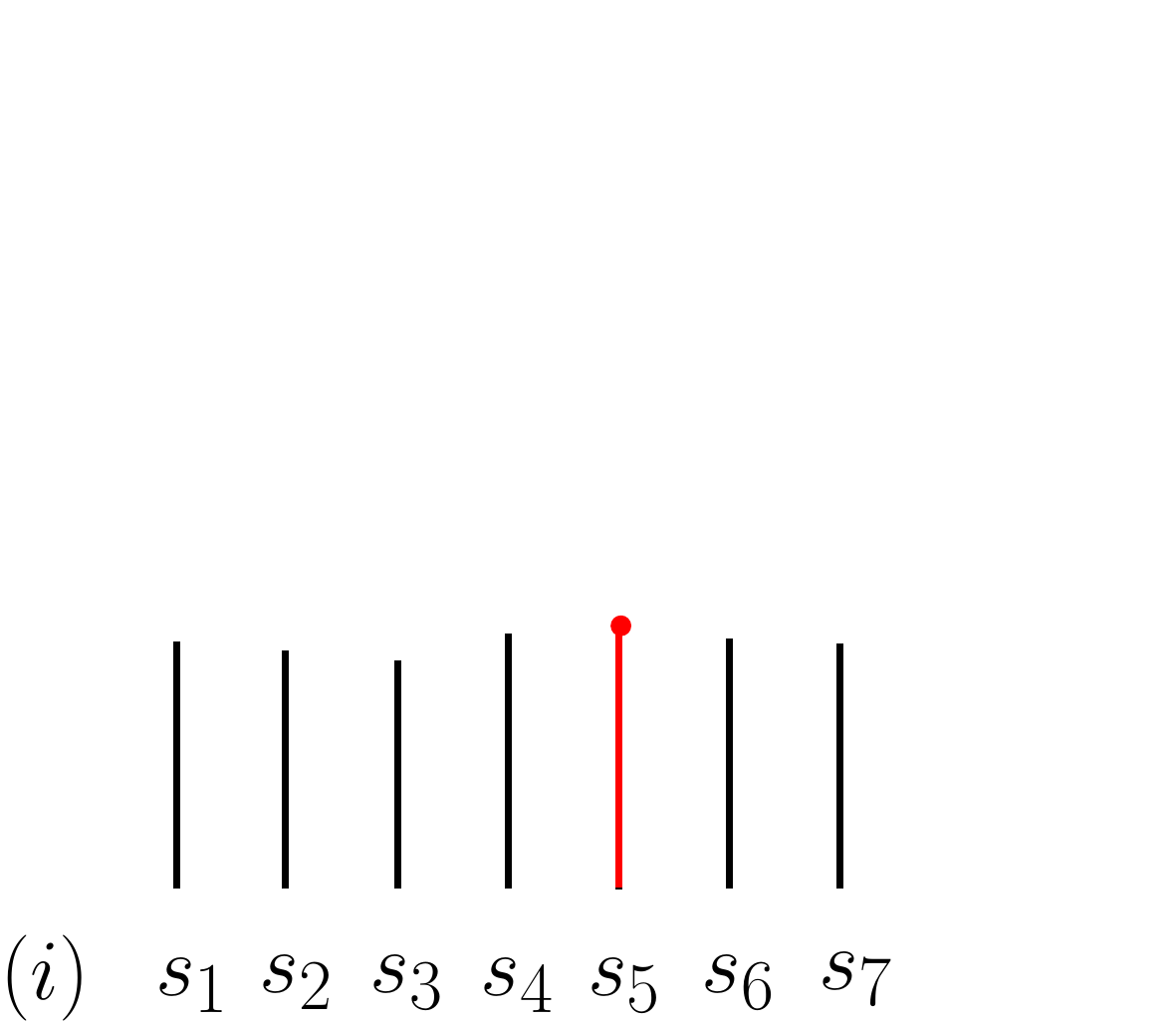}
\includegraphics[scale=0.33]{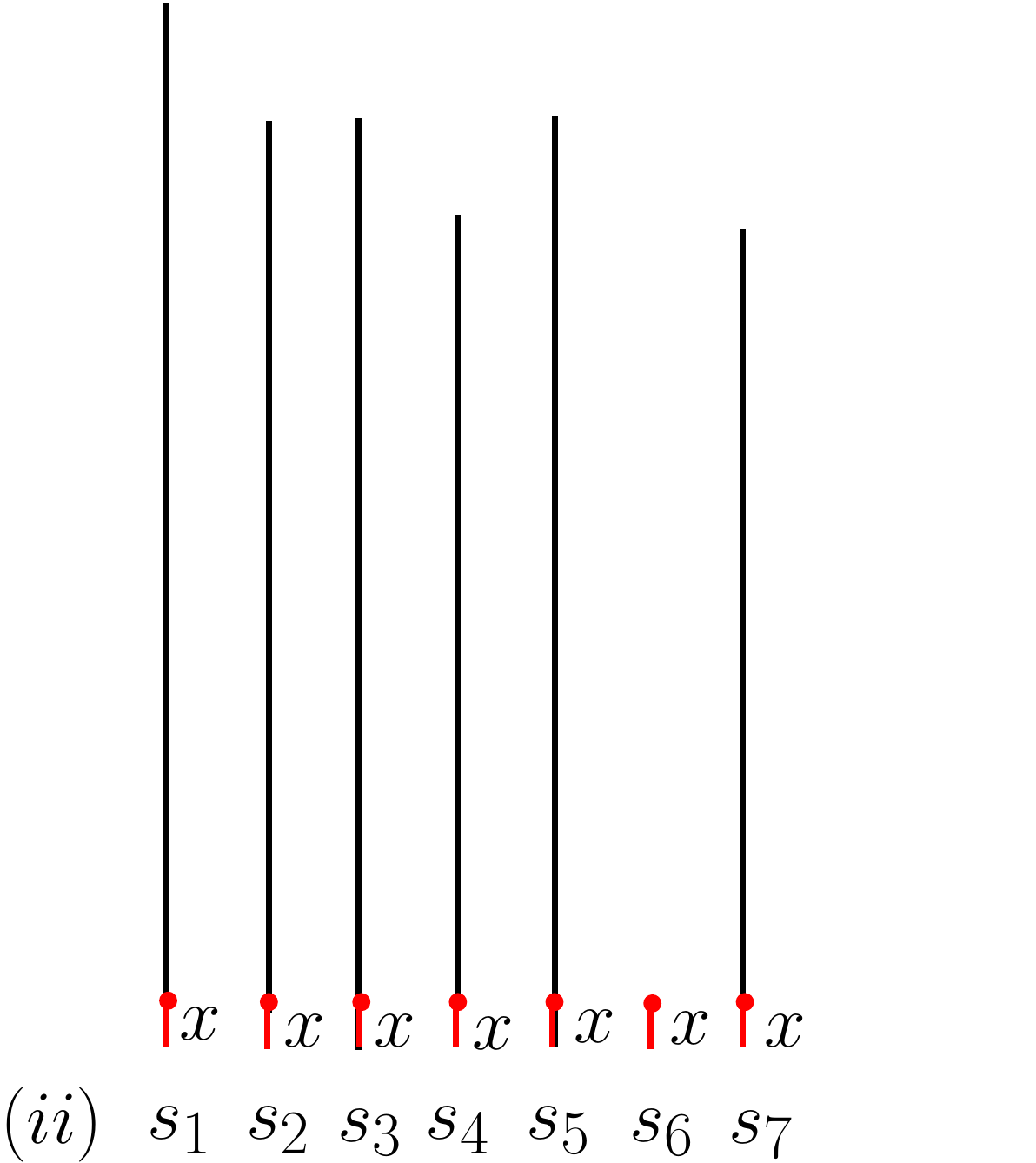}
\includegraphics[scale=0.33]{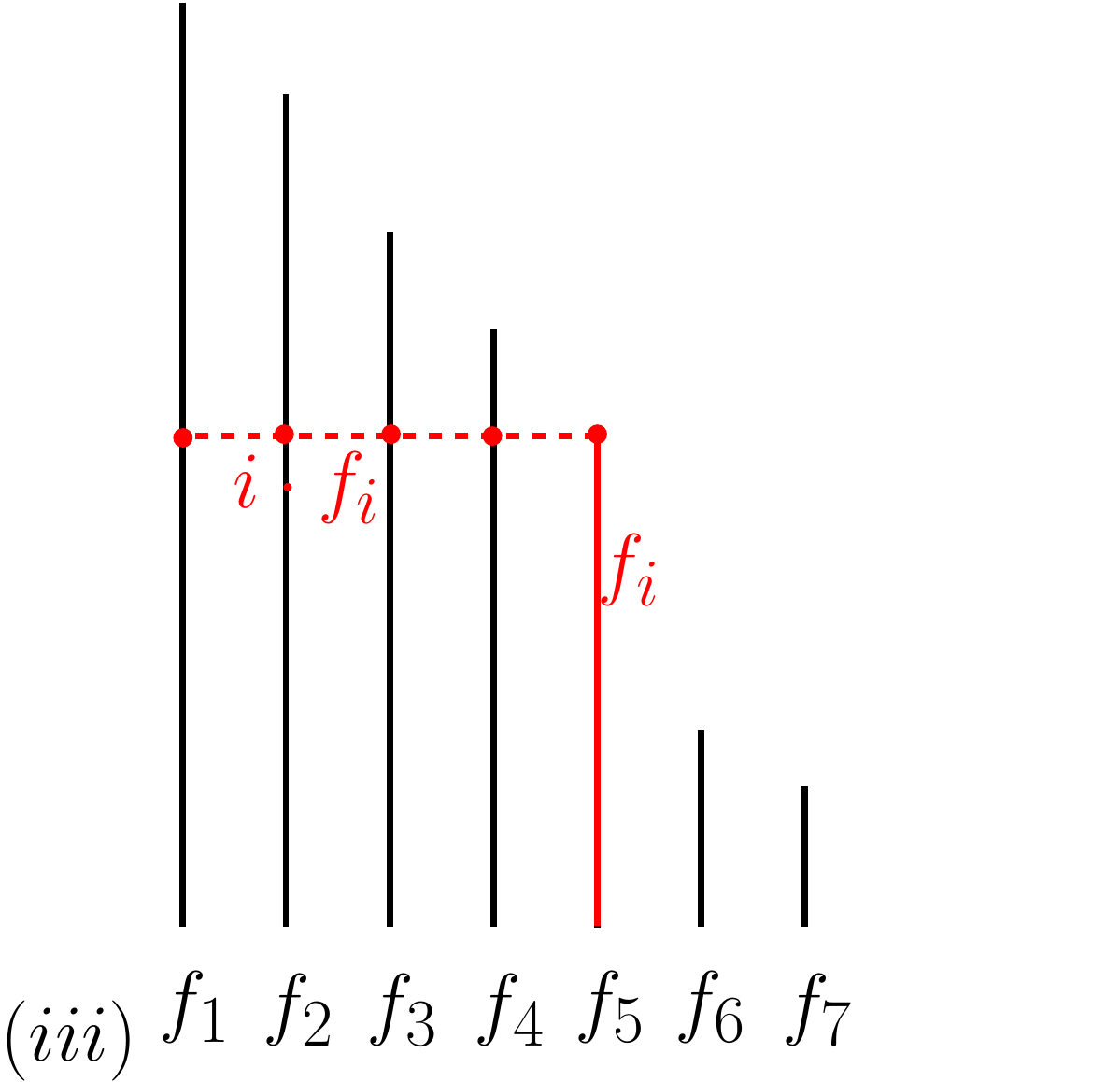}
\caption{
Online searching for the end of a segment (or digging for oil) for $m=7$ segments of unknown length.
There are two extreme cases:
(i) All segments have about the same length.
It is reasonable to move along an arbitrary segment up to the end, which is almost optimal.
(ii) One segment is significantly shorter than all other segments.
One will find the end of a shortest segment by checking all segments with its length.
(iii) In case the length of each segment is known, but not the corresponding number of segment.
There is always an optimal strategy:
Assume that $f_1\geq f_2 \geq \cdots \geq f_m$ is the decreasing order of the length of all segments.
An optimal strategy explores $i$ (arbitrary) segments up to depth $f_i$, where $i$ is chosen 
so that $i \cdot f_i =\min_{1 \leq k \leq m} k\cdot f_k$.}
\label{CorrFind-fig}
\end{center}
\end{figure}

\section{Online approximation of the certificate path\label{onlappe-sect}}
We are searching for a reasonable escape strategy in an unknown environment.
As shown in the previous section, the certificate path and its length is a reasonable candidate for comparisons.
Let us assume that $x(1+\alpha_x)$ is the length of the certificate for some polygon~$P$ and for 
an arbitrary distance $x$. We can assume that $\alpha_x\in [0,2\pi]$.
This holds since the shortest distance $d_s$ from $s$ to the boundary always results in a candidate $d_s(1+2\pi)$.
All other reasonable distances $x$ are larger than $d_s$ and $\alpha_x\leq 2\pi$ holds for the optimal~$x$. 

Similar to the considerations of Kirkpatrick (see Section~\ref{just-sect}), we would like to guarantee that 
we leave the polygon $P$ if we have overrun the distance $x$ more than $\alpha_x$ times.
This means that the boundary should not wind arbitrarily around $s$. 
Therefore, we restrict our consideration to a position $s$ in the kernel of a star-shaped polygon 
so that there is a single (unknown) distance to the outer boundary in any direction. 

In this case we apply the following logarithmic spiral strategy.  
A logarithmic spiral can be defined by polar coordinates $(\varphi,a\cdot e^{\varphi\cot(\beta)})$ 
for $\varphi\in (-\infty,\infty)$, a constant $a$ and an eccentricity $\beta$ as shown in Figure~\ref{SpiralStrat-fig}.
For an angle $\phi$, the path length of the spiral up to point $(\phi,a\cdot e^{\phi\cot(\beta)})$ is 
given by $\frac{a}{\cos\beta}e^{\phi\cot(\beta)}$. 
\begin{figure}[ht]
\begin{center}
\includegraphics[scale=0.4]{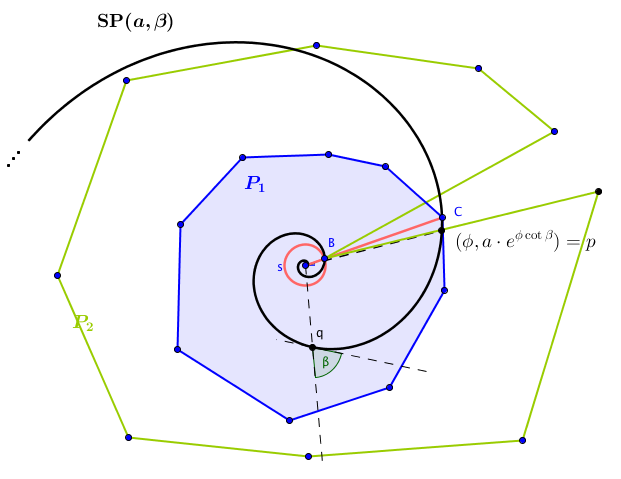}
\caption{We apply a spiral strategy for unknown polygons and an unknown starting point $s$ in the kernel.
The eccentricity $\beta$ is chosen so that the two extreme cases have the same ratio.
For both polygons $P_1$ and $P_2$, the strategy passes the boundary at point 
$p=(\phi,a\cdot e^{\phi\cot\beta})$ close to $C$.
The path length of the strategy for leaving the polygons is roughly the same. 
The certificate for $P_1$ has length $|s\;C|$ (checking the maximal distance to the boundary of $P_1$), 
whereas the certificate for $P_2$ has length $|s\;B|(1+2\pi$) with $|s\;C|=e^{2\pi\cot\beta} |s\;B|$ 
(checking the smallest distance to the boundary of $P_2$ with a full circle).
We can construct such examples for any point $p$ on the spiral.}
\label{SpiralStrat-fig}
\end{center}
\end{figure}
For our purpose we choose $\beta$ so that the two extreme cases of the certificate attain the same 
ratio; see Figure~\ref{SpiralStrat-fig}.
We can assume that the certificate of the environment is $x(1+\alpha_x)$ for an arbitrary distance $x$ 
and an angle $\alpha_x\in[0,2\pi]$. Since the spiral strategy checks the distances in a monotonically 
increasing and periodical way, there has to be some angle $\phi$ so that $x=a\cdot e^{(\phi-\alpha_x)\cot(\beta)}$ 
holds. This means that in the worst case, the spiral strategy will leave the environment at 
point $p=(\phi,a\cdot e^{\phi\cot(\beta)})$ with path length $\frac{a}{\cos\beta}  \cdot e^{\phi\cot(\beta)}$.
Exactly $\alpha_x$ distances of length $x$ have been exceeded, which means that the boundary 
has been reached. (Note that, this might not hold for points outside the kernel.)

We would like to choose $\beta$ so that the two extreme cases $\alpha_x=0$ and $\alpha_x=2\pi$ 
have the same ratio. Thus, we are searching for an angle $\beta$ so that 
\begin{eqnarray}
\frac{\frac{a}{\cos\beta} \cdot e^{\phi\cot \beta }}
{a\cdot e^{\phi\cot \beta }(1+0)} & = & \frac{\frac{a}{\cos\beta} \cdot e^{\phi\cot \beta }}
{a\cdot e^{(\phi-2\pi)\cot \beta}(1+2\pi)}\;\;\Leftrightarrow \label{Equal-equ1}\\
1 &=& \frac{e^{2\pi\cot \beta }}{1+ 2\pi }\label{Equal-equ2}
\end{eqnarray}
holds. 
The right-hand side of Equation~(\ref{Equal-equ1}) shows the case where $x_2=a\cdot e^{(\phi-2\pi)\cot(\beta)}$ and $\alpha_{x_2}=2\pi$ gives the certificate and 
the left-hand side shows the case that $x_1= a\cdot e^{\phi\cot(\beta)}$ and $\alpha_{x_1}=0$ gives 
the certificate $x_i(1+\alpha_{x_i})$, respectively.  In both cases the spiral will detect the boundary 
latest at point $p=(\phi,a\cdot e^{\phi\cot(\alpha)})$, because the spiral checks 
$2\pi$ distances larger than or equal to $x_2$ and at least one distance $x_1$. 
Figure~\ref{SpiralStrat-fig} 
shows the construction of  corresponding polygons $P_1$ and $P_2$.

The solution of Equation~(\ref{Equal-equ2}) gives 
$\beta=\mbox{arccot}\left(  \frac{\ln(2\pi +1)}{2\pi}\right) = 1.264714\ldots$ and 
the ratio is $\frac{1}{\cos\beta}= 3.3186738\ldots$. 
Fortunately, for all other values $x=a\cdot e^{(\phi-\gamma)\cot \beta}$ 
and $\alpha_x= \gamma$ for $\gamma\in (0,2\pi)$ the ratio is 
smaller than these two extremes. 
The overall ratio function is 
\begin{equation}\label{solutionOther}
f(\gamma)=\frac{\frac{a}{\cos\beta} \cdot e^{\phi\cot\beta}}
{a\cdot e^{(\phi-\gamma)\cot\beta}(1+\gamma)}= \frac{e^{\gamma\cot\beta}}{\cos\beta (1+ \gamma)}\mbox{ for } \gamma\in[0,2\pi]
\end{equation}
and Figure~\ref{PlotRatio-fig} shows the plot of all possible ratios of the spiral strategy with 
eccentricity $\beta$.
\begin{figure}
\begin{center}
\includegraphics[scale=0.325]{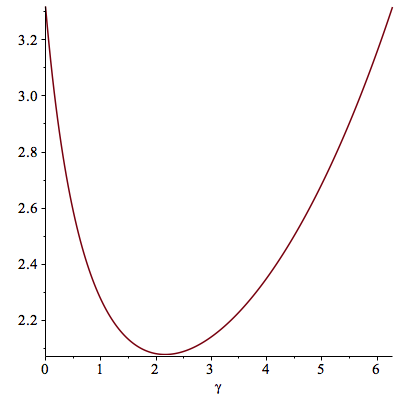}
\caption{The graph of the ratio function $f$ of Equation (\ref{solutionOther}) for the spiral strategy with 
eccentricity $\beta\approx 1.26471$. The two extreme cases $0$ and $2\pi$ have the same 
ratio $\approx 3.318674$ and all other ratios are strictly smaller.}
\label{PlotRatio-fig}
\end{center}
\end{figure}
Altogether, we have the following result.
\begin{theorem}\label{OptKernel-thm}
There is a spiral strategy for any unknown starting point $s$ inside 
the kernel of  an unknown environment $P$ that always hits the boundary with path length 
smaller than  $3.318674$ times the length of the corresponding certificate 
for $s$ and $P$. 
\end{theorem}
\begin{proof}
Assume that the certificate of $P$ and $s$ is given by $x(1+\alpha_x)$. 
We can set $\gamma:=\alpha_x$ and we will also find an angle $\phi$ so 
that $x=a\cdot e^{(\phi-\gamma)\cot \beta}$ holds. 
At point $p=(\phi,a\cdot e^{\phi\cot\beta})$ the spiral  has subsumed 
an arc of angle $\gamma$ with distances $x$, so the spiral strategy will leave $P$ 
at $p$ in the worst case.
(Note that, if the start point is not inside the kernel, this might not be true!)
The ratio is given by $f(\gamma)$ as in (\ref{solutionOther}) and 
Figure~\ref{PlotRatio-fig}. In the worst case for the strategy $\gamma$ is either $0$ 
or $2\pi$ for the ratio $3.318674$, respectively.
\end{proof}
We have designed a spiral strategy for some reasonable environments. 
In the next section, we give a lower bound that shows that this strategy is (almost) optimal 
for these environments. 

Note that a spiral strategy for the online approximation of the certificate path of an arbitrary 
unknown polygon and position cannot be competitive in general. The polygon might itself wind 
around the spiral. The ratio against the certificate might be arbitrarily large, consequently.
In more general environments  other online strategies have to be applied. 
A potential strategy might be a connected sequence of circles $C_i$ with exponentially 
increasing radii $r^i$. This online strategy should result in a constant competitive ratio.
Obtaining the optimal strategy for such cases might be complicated and gives rise to 
future work.

\section{Lower bound construction: Online strategy against the certificate}\label{opt-sect}
 \begin{theorem}\label{Opt2-thm}
Any strategy that escapes from an unknown environment $P$ in unknown 
position $s$ will achieve a competitive factor of at least $3.313126$  against the 
length of a corresponding certificate for $s$ and $P$ in the worst-case.
\end{theorem}
\begin{proof} 
Let us assume that a strategy $S$ is given that attains a better ratio $C$ in the worst case. 
We consider a bunch of $n$ rays emanating from $s$ with equidistant angle $\frac{2\pi}{n}$ 
as depicted in Figure~\ref{StratEx2-fig} for $n=8$. 

\begin{figure}
\includegraphics[width=0.499\textwidth,page=1]{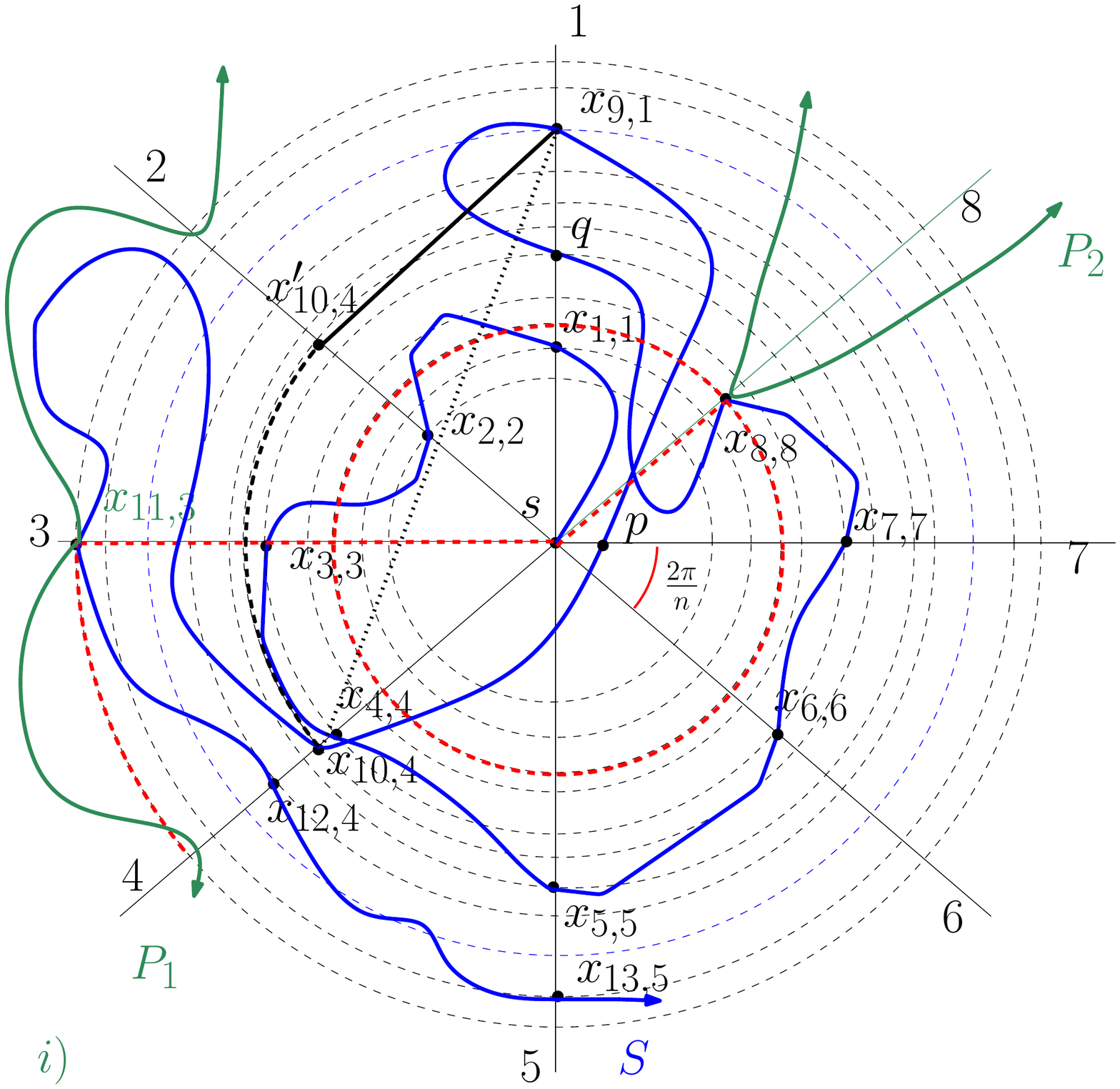}
\includegraphics[width=0.499\textwidth,page=2]{StratEx2.pdf}
\caption{i) The strategy $S$ results in a sequence $S'$ that represents the visits on $n$ 
rays successively. There will be a next entry $x_{i,j_i}$ if the strategy exceeds the distance 
on ray $j_i$. For two successive extensions on the same ray only the last entry is registered in $S'$. 
For  the subsequence $S'_{13}=(x_{1,1},x_{2,2},\ldots,x_{13,5})$ there will be a 
last visit on each ray, a minimal distance $x_m=x_{8,8}$ on ray $8$ and a maximal distance $x_M=x_{11,3}$
on ray $3$. 
These values gives rise to the construction of certificates for $S$ as sketched by the polygons $P_1$ and $P_2$. 
There are polygons $P_1$ and $P_2$ with certificates 
$x_M(1+\frac{2\pi}{n})$ and $x_m(1+2\pi)$ and so far $S$ has not been escaped 
from neither $P_1$ nor $P_2$. In $S'_{13}$ the direct distance between two successive points, 
for example $x_{9,1}$ and $x_{10,4}$, is shorter than the original path length on $S$ and 
we can further shorten the distance by assuming that $x_{10,4}$ is on the neighbouring ray 
as depicted by $x_{10,4}'$.
ii) We sort the entries of $S'_k$ into a sequence $X_k$ and visit the rays in 
increasing distance and periodic order. The path length of $X_k$ is not larger than the 
original path length of $S_k$ and the corresponding certificates for the maximal  and minimal value 
$x_M'=x_M$  and $x_m'>x_m$ are not smaller. Thus, the sum of the corresponding ratios gives a lower 
bound for the sum of the original ratios.
}
\label{StratEx2-fig}
\end{figure}

The strategy $S$ will successively extend the distances 
from $s$ also along the rays. Let the sequence $S'=(x_{1,j_1},x_{2,j_2},x_{3,j_3},\ldots)$ describe  
successive visits of the $n$ rays by the strategy. In $x_{i,j_i}$ the entry  $i$ 
stands for the order and $j_i$ stands for the ray.
In $S'$ we only register a visit on ray $j_i$, if it exceeds the previous 
visit on the ray $j_i$. Furthermore, if the distance at ray $j_i$ is exceeded in two successive entries 
we do not register the first visit in the sequence~$S'$.  

In Figure~\ref{StratEx2-fig}(i) we have registered $13$ successive visits $x_{i,j_i}$. Here for example 
the visit of ray $7$ at point $q$ between $x_{9,1}$ and $x_{10,4}$ was not registered in $S'$ because it 
does not improve the distance of the former visit $x_{7,7}$. Additionally, the visit of ray $1$ at point 
$q$ just before $x_{9,1}$ 
improved the distance $x_{1,1}$ but it was further improved on $x_{9,1}$  and in between 
no other ray was improved. For any continuous strategy $S$ we will find such an infinite sequence $S'$. 
Let $x_{i,j_i}$ denote the visit and also the distance to the starting point $s$. 

Let us assume that we stop the strategy $S$ at of some ray $j_k$ where the distance was
just exceeded on ray $j_k$, so $S'$ has $k$ steps. Let $S_k$ denote the sub-strategy of $S$ and 
$S'_k$ the corresponding subsequence.
There will be at least two 
ratios for $S_k$ that correspond to values of $S'_k$ as follows. 
In $S'_k$ we consider the last 
visits on each ray which gives the corresponding maximal visited distance to $s$ for each ray. 
There will be an overall maximal distance $x_M$ on some ray $M_j$ 
and a minimal distance $x_m$ on some other ray $m_j$. 
In Figure~\ref{StratEx2-fig}(i)  we have stopped the strategy $S$ at $x_{13,5}$ on ray $5$ and 
in $S_{13}'$ the minimal distance for the last \emph{round} is given by $x_m=x_{8,8}$ on ray $8$ 
and the maximal distance is given by $x_M=x_{11,3}$ on ray $3$. 

We can construct polygons $P_1$ and $P_2$ 
so that $x_M(1+\frac{2\pi}{n})$ is a certificate for $P_1$ and 
 $x_m(1+2\pi)$ is a certificate for $P_2$. In the first case all other rays have been visited 
 with depth smaller or equal to $x_M$ and we build a polygon $P_1$ outside $S_k$ 
 that visits any ray at $x_M-\epsilon$  and ray $M_j$ at $x_M$. 
 This means that a circular strategy with $x_M$ 
 and an arc of length $x_M\frac{2\pi}{n}$ will be sufficient and gives the
 certificate for $P_1$ (or at least an upper bound for the certificate of $P_1$). 
 See for example the polygon $P_1$ sketched in Figure~\ref{StratEx2-fig}(i)
 for the maximal visit $x_M=x_{13,5}$. On the other hand for the minimal value $x_m$ 
 we construct a polygon $P_2$ that hits $x_m$ on ray $m_j$ 
 but runs arbitrarily far away from  $S_k$ in any other direction. Thus, $x_m(1+2\pi)$ 
 gives the certificate (or at least an upper bound for the certificate of $P_2$). See for example the polygon 
 $P_2$ sketched in Figure~\ref{StratEx2-fig}(i) 
 for the minimal visit $x_m=x_{8,8}$.  We do not expect that $S_k$ has 
 already detected these polygons but $S$ finally will. So the ratio of the path length $|S_k|$
 over $x_m(1+2\pi)$   and also the ratio of the path length $|S_k|$ over $x_M(1+\frac{2\pi}{n})$ give 
 lower bounds  for the strategy $S$. Note that the half of the sum of the 
 two ratios  cannot exceed $C$ because otherwise at least one has to be greater than $C$. 
 
 For any such stop we will sort the values of  $S_k'$ in a sequence $X_k$ and we will visit the $n$ rays 
 in a monotone and periodic way by sequence $X_k$ connecting the points by
 line segments; see Figure~\ref{StratEx2-fig}(ii). We can prove that the overall path length of $X_k$ cannot 
 be larger than  $|S_k|$. 
 
 This can be seen as follows. Successively visiting the points of $S_k'$ in a polygonal chain is already a 
 short cut for $S_k$.  This polygonal path for $S_k'$ might move between
  two successive values $x_{i,j_i}$ to $x_{i+1,j_{i+1}}$ where $j_i$ and $j_{i+1}$ are not neighbouring rays. 
  In this case we can further short cut  the length of the chain of $S_k'$  by just counting a movement 
  from $x_{i,j_i}$ to the distance $x_{i+1,j_{i+1}}$ on one of the directly neighbouring rays. For example in 
  Figure~\ref{StratEx2-fig}(i)  the segment from $x_{9,1}$ and $x_{10,4}'$ improves 
  the path length from $x_{9,1}$ and $x_{10,4}$ but passes ray 2 and 3.  We only count the 
distance between $x_{9,1}$ and $x_{10,4}'$ on the neighbouring ray which further improves 
the length.   This means that 
 for a lower bound on the overall path length we can also consider a path that visits 
 two neighbouring rays with  angle $\frac{2\pi}{n}$ successively from one to the other with the 
 corresponding depth values $x_{i,j_i}$ stemming from $S_k'$. By triangle inequality it can be shown 
 that the shortest path that
 visit all the depth of a sequence $S_k'$ on two rays by changing from one 
 ray to the other in any step, visits the two rays successively 
 in an increasing order. A similar argument was applied by one author of this article  
 in~\cite{l-oss-10}  where a detailed proof of this property is given in the Appendix of~\cite{l-oss-10}. 
 Finally, we can rearrange the path of $S_k'$ to a path that visits the rays in a 
 periodic and monotone  way.
 
 Altogether, we have translated the strategy $S_k$ in a discrete strategy $X_k=(x_1,x_2,\ldots,x_k)$ with $k$ entries on
 $n$ rays that visit the rays in a periodic order such that $x_i$ visits ray $i \mbox{ mod }n$ and 
 with overall shorter path length; see Figure~\ref{StratEx2-fig}(ii). Consider the corresponding certificates of this new
 strategy in comparison to the original strategy.
 For the smallest value on the last round and the largest value on the last round 
we will obtain a certificate path $x_{k}(1+\frac{2\pi}{n})$ which is the same for the 
previous maximal value $x_M=x_k$ and a certificate path $x_{k-n+1}(1+2\pi)$ which is never smaller than 
$x_m(1+2\pi)$ for the minimal value  $x_m\leq x_{n-k+1}$.   The minimal value can only increase 
since we sorted the values of $S_k'$. For example in Figure~\ref{StratEx2-fig}(ii) 
we have $k=13$ and the minimal value in 
the last round is given by $x_6=x_{6,6}$ which is larger than $x_m=x_{8,8}$. 
Altogether, the sum of these two ratios in the periodic and monotone setting is always smaller than the 
sum of the ratios in the original setting. 

Finally, we would like to find a periodic and monotone strategy that optimizes the 
sum of exactly such ratios in this discrete version. This optimal strategy will perform at least as good
as any strategy obtained by the above reconstruction. Thus, the optimal value for the sum 
is a lower bound for the sum of two ratios in the original setting. 

For optimizing the sum for an arbitrary strategy we use an infinite sequence of values 
$X=(x_1,x_2,\ldots)$ and  we define the following functionals 
\begin{eqnarray}\label{Funk1-equ}
F^1_k(X) =\frac{\sum_{i=1}^{k-1}  \sqrt{ x_i^2-2\cos\left(\frac{2\pi}n\right) x_i x_{i+1}  +x_{i+1}^2}}{x_{k}(1+\frac{2\pi}{n})}
\end{eqnarray}
and 
\begin{eqnarray}\label{Funk2-equ}
F^2_k(X) =\frac{\sum_{i=1}^{k-1}  \sqrt{ x_i^2-2\cos\left(\frac{2\pi}n\right)x_ix_{i+1}  +x_{i+1}^2}}{x_{k-n+1}(1+2\pi)}
\end{eqnarray}
that represent the ratios. We are looking for a sequence $X$ so that 
$$ \inf_Y \sup_k F^1_k(Y)+F^2_k(Y) = D \mbox{ and }  \sup_k F^1_k(X)+F^2_k(X)=D$$ 
holds which shows that $D$ is the best sum ratio that we can achieve. 

Optimizing such discrete functionals can be done by the method proposed by Gal; 
see also Gal \cite{gc-oefmp-76,g-sg-80}, Alpern and Gal \cite{ag-tsgr-03}, 
and an adaption of Schuierer \cite{s-lbogs-01}. 
It is shown that under certain prerequisites there will be an optimal 
exponential strategy $x_i=a^i$. The main requirement is that 
the functional has to fulfil a unimodality property. This means that the 
piecewise sum of two strategies $X$ and $Y$  is never worse than 
one of the single strategies. This should also hold for a scalar multiplication of a 
single strategy. So any linear combination of strategies that are bounded by a constant 
will remain  bounded by the maximal bound. The proof of Gal shows that in this case 
we can always combine 
bounded strategies so that we 
finally get  arbitrarily close to an exponential strategy that has the same bound; 
see the full proof of Gal in~\cite{g-sg-80} Appendix 2, Theorem 1. 

We can easily show that the requirements for the main Theorem of Gal are fulfilled for 
both functionals $F^1_k(X)$  and $F^2_k(X)$. 
For a similar functional a detailed proof of this property 
was given in the Appendix of~\cite{l-oss-10}. 

Now let us assume that we have an 
optimal strategy $X$ for the sum, say
$F^1_k(X)+F^2_k(X)$. This means that both functionals will also be bounded 
by constants $D_1$ and $D_2$ w.r.t. $X$. 
We make use of linear combination of $X$ but apply them 
independently to the functionals $F^1_k(X)$ and $F^2_k(X)$. The Theorem of Gal 
shows that we will get arbitrarily close to an exponential strategy $x_i=a^i$ that 
is not worse than $X$ for both $F^1_k(X)$ and $F^2_k(X)$. This means 
that $x_i=a^i$ is also not worse than $X$ for the sum functional. 

Altogether, it is allowed to search for the best strategy $x_i=a^i$ and we have to 
optimize 

\begin{eqnarray}
\frac{\sum_{i=1}^{k-1}  \sqrt{ a^{2i}-2\cos\left(\frac{2\pi}n\right)a^{2i+1}  +a^{2i+2}}}{a^k(1+\frac{2\pi}{n})} 
& + & 
\frac{\sum_{i=1}^{k-1}  \sqrt{ a^{2i}-2\cos \left(\frac{2\pi}n\right)a^{2i+1}  +a^{2i+2}}}{a^{k-n+1}(1+2\pi)}\nonumber\\[2ex]
\Leftrightarrow\sum_{i=1}^{k-1} a^i \left(\frac{\sqrt{ 1-2\cos \left(\frac{2\pi}n\right) a  +a^{2}}}{a^k(1+\frac{2\pi}{n})}\right)& + & 
\sum_{i=1}^{k-1}  a^i \left(\frac{\sqrt{ 1-2\cos \left(\frac{2\pi}n\right)a  +a^{2}}}{a^{k-n+1}(1+2\pi)}\right). \label{FirstRatio}
\end{eqnarray}

For Equation~(\ref{FirstRatio}) we resolve the  geometric serie part and simplify the expression to the 
minimization of 
\begin{eqnarray}\label{LastRatio}
g_n(a):=\frac{1}{a-1}\left(\frac{\sqrt{ 1-2\cos \left(\frac{2\pi}n\right) a  +a^{2}}}{(1+\frac{2\pi}{n})}\right) +
\frac{a^{n+1}}{a-1}\left(\frac{\sqrt{ 1-2\cos \left(\frac{2\pi}n\right)a  +a^{2}}}{(1+2\pi)}\right)\;.
\end{eqnarray} 
We minimize Equation~(\ref{FirstRatio}) by numerical means. 
For any number of rays $n$ a minimal value of $g_n(a)$ gives a lower bound on the sum of two 
ratios in the original problem. So we can choose $n$ as large as we want. 
We minimize $g_n(a)$ by numerical means using Maple. 
For example for $n=28000000000$ we obtain $a=1.0000000006809\ldots$ and 
$g(a)=6.62521\ldots$ 
This means that for an arbitrary strategy of the original problem there will always be at least one 
ratio larger than $\frac{6.6252}{2}=3.313126$ which finishes the proof. 
\end{proof}

\section{Conclusion\label{concl-sect}}
We have introduced a new, simple and intuitive performance measure for the 
comparison against an online escape path for an unknown environment. 
The measure outperforms the (few) known ultimate optimal escape paths 
of convex environments and is also sort of a generalization of a discrete list searching approach by Kirkpatrick. 

For a more general class of environments, we presented an online spiral strategy that approximates 
the measure within an (almost) optimal factor of $\approx3.318674$. 
Different to classical results, the spiral optimizes against two extremes. 
It was shown that the factor is almost tight by constructing a lower bound that also holds for 
arbitrary environments.
This is one of the very few cases, where the optimality of spiral search is verified. 

Future work might consider randomization. Additionally, it will be helpful to prove the
strong conjecture that the certificate path is indeed 
always better than the shortest escape path for all environments (even when the 
best the escape path is not known).  

\paragraph*{Acknowledgements:}
We would like to thank all anonymous referees for their helpful comments 
and suggestions.

\appendix
\section{Appendix\label{conf-app}}
\setcounter{theorem}{0}

\subsection{Efficient computation of the certificate} \label{appendix:polytimecomp-sect}

\newcommand{\certPath}{\ensuremath{\Pi_s}\xspace}
\newcommand{\pathArc}{\ensuremath{\alpha_x}\xspace}
\newcommand{\pathSeg}{\ensuremath{\pi_{seg}}\xspace}
\newcommand{\schinzel}[2]{\ensuremath{\lambda_{#1}(#2)}}
\newcommand{\radius}{\ensuremath{x}\xspace}
\newcommand{\sweepCircle}{\ensuremath{C}\xspace}
\newcommand{\region}{\ensuremath{P}\xspace}

In the following, we describe an efficient algorithm to compute \certPath for $s\in \region$.
We assume that \region is bounded by $n$ line-segments.
Each segment is defined by two points $p_i, p_{i+1}$ in the plane.
Thus, the overall number of points is $n+1$.
The algorithm proceeds in two phases.
At first, \textit{candidates} for \certPath are computed.
For a given radius \radius, those candidates have a maximum arc \pathArc in the corresponding direction.
The length of such a candidate can be described as a function of the radius \radius.
In a second phase, an efficient algorithm for computing the upper envelope is applied to these functions to obtain the overall shortest path \certPath.

For reasons of simplicity, we assume that a candidate path intersects with the border of \region two times.
If \pathArc equals zero or $2\pi$, both intersections are identical.
Since we know that a candidate can be described by two segments, we consider its length.
Each segment $s_i$ is part of the line $l_i: \lbrace p_i + t \cdot\overrightarrow{p_i p_{i+1}} | t\in\mathbb{R} \rbrace$.
We denote $d_i$ the shortest distance from $s$ to a segment $s_i$ and $\measuredangle l_i l_j$ the positive angle between both line segments in the direction of $s$.
Moreover, we denote $\left[x^-; x^+ \right]\subset\mathbb{R}$ the existence interval for which the candidate path actually intersects both segments on $l_i$ and $l_j$ at the same time.
This allows us to express \pathArc as a function of \radius
\begin{equation*}
\pathArc = \pi - \measuredangle l_i l_j \pm \left( \arccos\left(\frac{d_i}{r}\right) + \arccos\left(\frac{d_j}{r} \right)\right)
\qquad \text{ for } \radius\in\left[ x^-; x^+ \right].
\end{equation*}
There will not be a candidate path for every pair of segments.
As we show in the following, the number of such paths is linear, indeed.
We also show how all candidate paths can be computed efficiently.

\begin{lemma}\label{lemma:numberOfCandidates}
All candidates for the shortest certificate path can be constructed in \bigO{n \cdot \log n} deterministic time.
The overall number is bounded by \bigO{n}.
\end{lemma}

\begin{proof}
To compute all candidates, we sweep out the region \region with a circle \sweepCircle of expanding radius \radius around $s$.
Initially, we consider the candidate path touching the segment closest to $s$.
This path always exists and has $\pathArc = 2\pi$.
In the course of the sweep algorithm, the radius \radius of \sweepCircle expands.
Until no other point of the bound of \region lies on or in \circle, the initial candidate path is the only one.
However, when the circle reaches a segment, the current candidate path also reaches its corresponding existence interval.
Depending on how \sweepCircle meets the segment, there are three possible type of events; see Figure \ref{figure:event} on the left.
\begin{enumerate}
\item
\sweepCircle meets $p_i$ and both neighbouring segments (including $p_{i-1}$ and $p_{i+1}$) lie inside \sweepCircle.
Then, the current candidate path in this direction degenerates to a straight line.
We can set $x^+$ of the candidate path to the current radius of \sweepCircle.
The total number of candidates does not change for this event.
\item
It may also be the case that \sweepCircle meets $p_i$ so that exactly one of the neighbouring segments lies inside (while the other outside) the circle.
Again, the candidate path reaches $x^+$, as he will no longer intersect with the inner segment.
However, there emerges a (single) new candidate path with one intersection on the segment $\left(p_i, p_{i+1} \right)$ and the second intersection with the same segment as the current candidate.
For this new candidate, we set $x^-$ to the current radius of \sweepCircle.
Obviously, the total number of candidates increases by one.
\item
The third type of event embraces two cases.
Either \sweepCircle meets $p_i$ so that both neighbouring segments lie outside the circle.
Or the circle gets tangential to a segment $\left(p_i, p_{i+1}\right)$.
Both cases involve the creation of two new candidate paths and the limitation of the existence interval of the candidate path in the current direction.
The total number of candidates increases by exactly two.
\end{enumerate}
Each type of event involves that one point (or segment) lies on (gets tangential to) the expanding sweep circle \sweepCircle.
Afterwards, the point lies inside \sweepCircle - a segment will never get tangential to \sweepCircle, again.
As the number of points and segments that define \region is linear, the number of events is \bigO{n} for each type of event and all together.
This has two consequences.

On the one hand, the sweep algorithm terminates after all \bigO{n} events have been handled.
No additional events are created through the course of the algorithm.
This gives the running time, as the events have to be handled ordered by \radius increasing.

On the other hand, each event involves the creation of at most two additional paths; see Figure \ref{figure:event}.
Thus, the overall number of candidate paths is bounded by \bigO{n}.
\begin{figure}
	\begin{minipage}{0.54\textwidth}
		\flushleft
		\includegraphics[width=\textwidth]{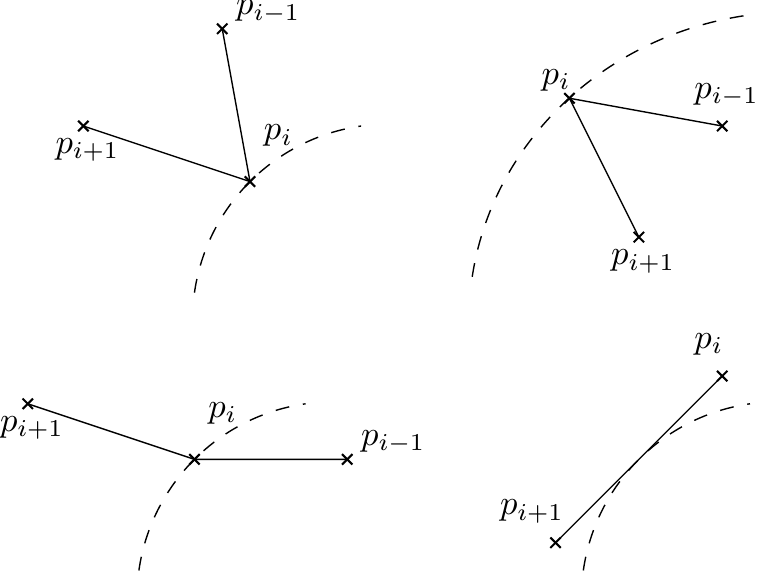}
	\end{minipage}
	\begin{minipage}{0.44\textwidth}
		\flushright
		\includegraphics[width=0.75\textwidth]{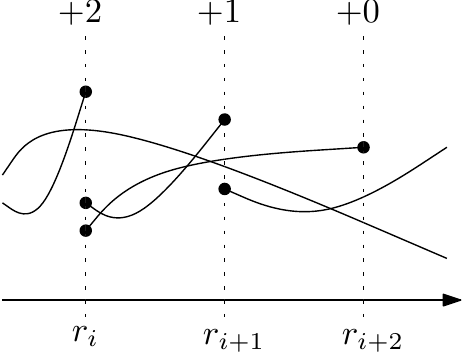}	
	\end{minipage}
	\caption{
	The figure of on the left shows the different type of events, which may occur when the dashed sweep touches a segment of $R$.
	The figure on the right plots the length of several candidate paths as functions of the radius \radius.
	At each event $r_i$, the overall number of candidate paths increases at most by two.}
	\label{figure:event}
\end{figure}
\end{proof}
The length of a candidate path defines a function of a real value \radius on an existence interval $\left[ \radius^-; \radius^+ \right]$.
We extend each of these functions to functions over $\mathbb{R}$ by adding a half-line to each end.
These half-lines can be chosen to run in parallel to all other additional half-lines.
As each length function has a unique maximum or minimum, we can choose the half-lines in a way that they do not intersect more than once with another function.
Then, we compute the upper envelope of no more than $O(n)$ curves.
Two such curves can intersect only $t$ times for some constant $t$. 
Thus, the overall upper envelope of all curves has complexity $O(\lambda_{t+2}(n))$ (almost linear!) and can be computed in $O(\lambda_{t+2}(n)\log n)$; see also Sharir and Agarwal~\cite{sa-dsstg-95}.  

Having computed the upper envelope of all segments, we check each part of the envelope for a minimum.
After \bigO{\schinzel{t+2}{n}} such checks, we obtain 
the certificate path. 

\begin{lemma}
The certificate \certPath of a polygon \region with $n$ edges and a starting point $s\in \region$ can be computed in $O(\lambda_{t+2}(n)\log n)$ time for some constant $t$. 
\end{lemma}


\end{document}